\documentclass[12pt,reqno]{amsart}

\usepackage{amsmath,amsfonts,amsthm,amssymb,amsxtra
}
\usepackage{bbm} 
\usepackage{hyperref} 

\newtheorem{theorem}{Theorem}
\newtheorem{proposition}[theorem]{Proposition}
\newtheorem{lemma}[theorem]{Lemma}

\theoremstyle{definition}

\theoremstyle{remark}

\newtheorem{remark}[theorem]{Remark}
\newtheorem{remarks}[theorem]{Remarks}

\newcommand{\C}{\mathbb{C}}
\newcommand\eps\varepsilon
\newcommand{\R}{\mathbb{R}}

\renewcommand\Re{\mathop{\mathrm{Re}}\nolimits}
\newcommand\aj{a^{(j)}}

\DeclareMathOperator{\Tr}{Tr}

\begin{document}

\title[The generalized Wehrl entropy bound]{The generalized Wehrl entropy bound\\ in quantitative form}

\author{Rupert L. Frank}
\address[Rupert L. Frank]{Mathe\-matisches Institut, Ludwig-Maximilans Universit\"at M\"unchen, The\-resienstr.~39, 80333 M\"unchen, Germany, and Munich Center for Quantum Science and Technology, Schel\-ling\-str.~4, 80799 M\"unchen, Germany, and Mathematics 253-37, Caltech, Pasa\-de\-na, CA 91125, USA}
\email{r.frank@lmu.de}

\author{Fabio Nicola}
\address[Fabio Nicola]{Dipartimento di Scienze Matematiche, Politecnico di Torino, Corso Duca degli Abruzzi 24, 10129 Torino, Italy}
\email{fabio.nicola@polito.it}

\author{Paolo Tilli}
\address[Paolo Tilli]{Dipartimento di Scienze Matematiche, Politecnico di Torino, Corso Duca degli Abruzzi 24, 10129 Torino, Italy}
\email{paolo.tilli@polito.it}

\renewcommand{\thefootnote}{${}$} \footnotetext{\copyright\, 2023 by the authors. This paper may be reproduced, in its entirety, for non-commercial purposes.\\
Partial support through US National Science Foundation grant DMS-1954995 (R.L.F.), as well as through the Deutsche Forschungsgemeinschaft (DFG, German Research Foundation) through Germany’s Excellence Strategy EXC-2111-390814868 (R.L.F.) is acknowledged.}

\begin{abstract}
	Lieb and Carlen have shown that mixed states with minimal Wehrl entropy are coherent states. We prove that mixed states with almost minimal Wehrl entropy are almost coherent states. This is proved in a quantitative sense where both the norm and the exponent are optimal and the constant is explicit. We prove a similar bound for generalized Wehrl entropies. As an application, a sharp quantitative form of the log-Sobolev inequality for functions in the Fock space is provided. 
\end{abstract}

\maketitle

\section{Introduction and main result}

We consider the following $L^2(\R)$-normalized Gaussian
$$
\phi(x) := 2^{1/4} e^{-\pi x^2} \,,
\qquad\text{for all}\ x\in\R \,,
$$
as well as its translations and modulations, for $(x_0,\omega_0)\in\R^2$,
$$
\phi_{(x_0,\omega_0)}(x) = e^{2\pi i\omega_0x} \phi(x-x_0) 
\qquad\text{for all}\ x\in\R \,.
$$
In various contexts in mathematics and its applications one arrives at the operator $\mathcal V$ that transforms a function $f\in L^2(\R)$ into a function $\mathcal Vf$ on $\R^2$, defined by
$$
(\mathcal V f)(x_0,\omega_0) = \int_{\R} \overline{\phi_{(x_0,\omega_0)}(x)} f(x)\,dx \,.
$$
The operator $\mathcal V$ is known as coherent state transform in mathematical physics and quantum mechanics and as short-time Fourier transform in signal processing and time-frequency analysis. It is easy to see that $\mathcal Vf\in L^2(\R^2)$ with
$$
\iint_{\R^2} |\mathcal Vf(x_0,\omega_0)|^2 \,dx_0\,d\omega_0 = \int_\R |f(x)|^2\,dx \,,
$$
so for $L^2$-normalized $f$, $|\mathcal Vf|^2$ can be viewed as a probability density on $\R^2$.

More generally, consider a ``density matrix", namely an operator $\rho\geq0$ on $L^2(\mathbb{R})$ such that ${\rm Tr}\, \rho=1$. The function
\[
u_\rho(x,\omega):=\langle \phi_{(x,\omega)},\rho \phi_{(x,\omega)}\rangle 
\]
is known as the Husimi function, covariant, or lower symbol of $\rho$. It is still a probability density. Moreover if $\rho$ is the orthogonal projection on some $L^2(\mathbb{R})$-normalized function $f$ - we write $\rho=|f\rangle \langle f|$ - we have $u_\rho(x,\omega)=|\mathcal{V} f(x,\omega)|^2$. \par\medskip

In the late 1970s Wehrl \cite{We} suggested a certain quantity, based on the covariant symbol of $\rho$, as a measure of the entropy of a quantum state and demonstrated several interesting properties of it. We refer to \cite{Sc} for a review of this topic. Precisely, Wehrl's entropy of the density matrix $\rho$ is defined as
$$
- \iint_{\R^2} u_\rho(x_0,\omega_0) \ln u_\rho(x_0,\omega_0)\,dx_0\,d\omega_0 \,,
$$
that is the continuous entropy of the probability density $u_\rho$.

In 1978 Lieb \cite{Li} proved Wehrl's conjecture that this entropy assumes its minimal value for $\rho=|\phi_{(x_0,\omega_0)}\rangle \langle \phi_{(x_0,\omega_0)}|$ where $(x_0,\omega_0)\in\R^2$ is arbitrary. Lieb deduced this in an ingenious fashion from the sharp forms of the Hausdorff--Young and Young inequalities. Later, Carlen \cite{Ca} provided a new proof of Lieb's result and showed that the above pure states are the only ones for which the minimal value is attained.

Our main theorem in this paper is a quantitative version of this result. Specifically, we will show that if $\rho$ has almost minimal Wehrl entropy, then it is close to the above family of pure states in the sense that its squared trace distance to this family of states is bounded from above by the discrepancy in entropy. We set
\begin{equation}\label{eq drho}
D[\rho] := \inf_{x_0, \omega_0\in\R} \|\rho-|\phi_{(x_0,\omega_0)}\rangle \langle \phi_{(x_0,\omega_0)}|\|_{S_1},
\end{equation}
where $\|\cdot\|_{S_1}$ denotes the trace norm.

\begin{theorem}\label{wehrlintro}
	There is a $c_*>0$ such that, for all operators $\rho\geq 0$ on $L^2(\R)$ with ${\rm Tr}\,\rho=1$,
	\begin{align}\label{eq:wehrlintro}
		- \iint_{\R^2} u_\rho(x,\omega) \ln u_\rho(x,\omega) \,dx\,d\omega
		& \geq - \iint_{\R^2} |\mathcal V\phi(x,\omega)|^2\ln|\mathcal V\phi(x,\omega)|^2 \,dx\,d\omega \notag \\
		& \quad + c_* D[\rho]^2 \,.
	\end{align}
\end{theorem}

\begin{remarks}
	(a) The constant $c_*$ in this theorem is explicit (in the sense that it is not obtained by a compactness argument).\\
	(b) The power $2$ of $D[\rho]$ in the theorem is optimal, as is easily seen by taking $\rho$ as a small perturbation of $|\phi\rangle \langle \phi|$.\\
	(c) One easily computes $- \iint_{\R^2} |\mathcal V\phi(x,\omega)|^2\ln|\mathcal V\phi(x,\omega)|^2 \,dx\,d\omega= 1$.
\end{remarks}

Theorem \ref{wehrlintro} will be a special case of a more general result. To motivate it, let us recall that in his proof of Wehrl's conjecture, Lieb proceeded by considering
\begin{equation}
	\label{eq:wehrlgen}
	\iint_{\R^2} \Phi(u_\rho(x,\omega)) \,dx\,d\omega
\end{equation}
with $\Phi(u)=u^r$ for $r\geq 1$ and showing that, among $\rho\geq 0$ with $\Tr\rho=1$, this quantity is maximal for $\rho=|f\rangle\langle f|$ with $f=\phi_{(x_0,\omega_0)}$. Many years later, Lieb and Solovej \cite{LiSo1} showed that the same is valid for any convex and continuous function $\Phi:[0,1]\to\R$. They deduced this generalized Wehrl conjecture via a limiting argument from their resolution of the analogue of Wehrl's conjecture for coherent spin states; see also \cite{LiSo2} for an alternative proof of the latter. Another proof of the generalized Wehrl conjecture was given in \cite{GiHoMa}, based on certain deep facts on Gaussian channels in quantum information theory. The fact that \eqref{eq:wehrlgen} is maximized \emph{only} for the above family of extremals (provided $\Phi$ is not affine linear in the case of pure states, or strictly convex in the general case) was shown recently in \cite{Fr,KuNiOCTi}.

Our main result is a quantitative version of this generalized Wehrl conjecture.

\begin{theorem}\label{thm1}
	Let $\Phi:[0,1]\to\R$ be continuous and convex with $\Phi(0)=0$ and assume that $\Phi$ is not a linear function. Then there is a constant $c_\Phi>0$ such that, for all operators $\rho\geq 0$ on $L^2(\R)$ with ${\rm Tr}\, \rho=1$,
	\begin{equation}\label{SF}
		\iint_{\mathbb{R}^2} \Phi(u_\rho(x,\omega))\, dx d\omega\leq \iint_{\mathbb{R}^2} \Phi(|\mathcal{V}\varphi(x,\omega)|^2)\, dx d\omega-c_\Phi D[\rho]^2. 
	\end{equation}
\end{theorem}

\begin{remarks}
	(a) As in Theorem \ref{wehrlintro}, the constant $c_\Phi$ is explicit (in the sense that it is not obtained by a compactness argument).\\
	(b) The power $2$ of $D[\rho]$ in the theorem is optimal, as is easily seen by taking $\rho$ as a small perturbation of $|\phi\rangle \langle \phi|$.\\ 
	(c) One easily computes $\iint_{\R^2} \Phi(|\mathcal V\phi(x,\omega)|^2)\,dx\,d\omega= \int_0^1 \Phi(\rho)\,\frac{d\rho}{\rho}$.\\
	(d) For $\Phi(\rho)= \rho\ln\rho$ one obtains Theorem \ref{wehrlintro}.\\
	(e) The quantity $\iint_{\mathbb{R}^2} \Phi(u_\rho)\, dx d\omega$ is therefore maximized \textit{only} by the coherent states provided $\Phi$ is not linear (not necessarily strictly convex).  
\end{remarks}

The developments that made the characterization of cases of equality in the generalized Wehrl conjecture possible \cite{Fr,KuNiOCTi} and that are also at the basis of our analysis here have their roots in the paper \cite{NiTi}, where an optimal Faber--Krahn inequality for the short-time Fourier transform/the coherent state transform was proved. As an aside we mention the paper \cite{Ku} by Kulikov that modified the method of \cite{NiTi} to prove contractivity properties of embeddings in Hardy and Bergman spaces and to verify an analogue of Wehrl's conjecture for the affine-linear group \cite{LiSo3}.

Very recently, in \cite{GGRT} a quantitative version of this optimal Faber--Krahn inequality was obtained. The present paper is a companion to this result.

In fact, we will show two approaches to Theorem \ref{thm1}. 
We first give a `direct proof', that makes the analysis in \cite{NiTi} quantitative. At a crucial point in the analysis a nontrivial generalization of a lemma from the proof of the quantitative Faber--Krahn inequality will enter; see Lemma \ref{lemma Paolo} below. Then we give a second proof of Theorem \ref{thm1}, in the special case of pure states ($\rho=|f\rangle\langle f|$), that takes the quantitative version of the Faber--Krahn inequality as a black box and uses it to deduce the quantitative version of the generalized Wehrl entropy bound. We believe that both proofs have their value and we present them side by side. 

For simplicity we confine ourselves to the one-dimensional case, but the arguments below can be easily
generalized to higher dimensions (see Remark \ref{multidim}).

As an application of Theorem \ref{wehrlintro}, in Section \ref{sec log-sobolev} we will prove a sharp quantitative form of the log-Sobolev inequality for functions in the Fock space.


\section{A first proof of the main result}
Let $\rho$ be a ``density matrix" i.e. an operator $\rho\geq0$ on $L^2(\R)$ such that ${\rm Tr}\, \rho=1$. By the spectral theorem we can write 
\begin{equation}\label{eq fdec}
\rho=\sum_j p_j |f_j\rangle \langle f_j |\qquad{\rm with}\ p_j\geq 0,\ \sum_j p_j=1,\ \langle f_j,f_k\rangle =\delta_{j,k}. 
\end{equation}
We have a corresponding decomposition for its covariant symbol
\begin{equation}\label{eq fcov}
u_\rho(x,\omega):=\langle \varphi_{(x,\omega)}, \rho \varphi_{(x,\omega)}\rangle= \sum_j p_j |\langle f_j,\varphi_{(x,\omega)}\rangle|^2=\sum_j p_j|\mathcal{V} f_j(x,\omega)|^2. 
\end{equation}
By using the relationship between the short-time Fourier transform and the Bargmann transform $\mathcal{B}$ \cite[Proposition 3.4.1]{grochenig}, we see that, identifying $\mathbb{R}^2$ with $\mathbb{C}$ and writing $z=x+i\omega$, we have
\begin{equation}\label{eq ffock}
u_\rho(z)=\sum_j p_j |F_j(z)|^2 e^{-\pi |z|^2},
\end{equation}
where $F_j(z)=\overline{\mathcal{B}f_j(\overline{z})}$ are entire functions, normalized in the Fock space $\mathcal{F}^2$, namely 
\[
\|F_j\|^2_{\mathcal{F}^2}:=\int_{\mathbb{C}} |F_j(z)|^2 e^{-\pi |z|^2} dA(z)=1,
\]
 with $dA(z)=dx\, dy$. And since the monomials $\bigl\{\pi^{n/2}z^n/\sqrt{n!}\bigr\}$ ($n\geq 0$)
form an orthonormal basis of the Fock space, each $F_j$ can be expanded as
    \begin{equation}
    \label{normaliz}
    F_j(z)=\sum_{n\geq 0} \aj_n \frac {\pi^{n/2} z^n}{\sqrt{n!}},\quad
    \text{where}\quad
    \sum_{n\geq 0} |\aj_n|^2=1.
    \end{equation}

 We can now state some basic property of the function $u_\rho(z)$ in \eqref{eq fcov}.
\begin{lemma}\label{lem pre}
Let $\rho$ be a density matrix as in \eqref{eq fdec}. Its covariant symbol $u_\rho(z)$ in \eqref{eq fcov} is a nonnegative real-analytic function in $\mathbb{R}^2$ vanishing at infinity. 
\end{lemma} 
\begin{proof}
The series $\sum_j p_j |\langle f_j,\varphi_{(x,\omega)}\rangle|^2$ converges uniformly in $\mathbb{R}^2$, because  $|\langle f_j,\varphi_{(x,\omega)}\rangle|\leq 1$ by the Cauchy--Schwarz inequality. Hence $u$ vanishes at infinity, since each function $\langle f_j,\varphi_{(x,\omega)}\rangle$ vanishes at infinity. 

To prove that $u$ is real-analytic, we use the representation \eqref{eq ffock}. Since $|F_j|^2$ is subharmonic and $\|F_j\|^2_{\mathcal{F}^2}=1$ we have the pointwise estimate (cf. \cite[Theorem 2.7]{Zhu})
\[
|F_j(z)|\leq e^{\pi |z|^2}.
\]
Hence by the Cauchy estimates we see that for every $R$ there exists $C_R>0$ such that, for every $h,k\geq0$, 
\[
|\partial^h \bar{\partial}^k |F_j(z)|^2|=|\partial^{h+k} F(z) |\leq C_R^{h+k+1} (h+k)!,\qquad |z|\leq R. 
\]
Since $C_R$ is independent of $j$, this implies that the series in \eqref{eq ffock} converges in $C^\infty(\mathbb{R}^2)$ and $u$ is real-analytic. 
\end{proof}
Given a density matrix $\rho$, and its covariant symbol $u_\rho$, throughout this section we set 
\begin{equation}\label{eq fT}
T:=\max_{(x,\omega)\in\mathbb{R}^2} u_\rho (x,\omega).
\end{equation}
Clearly $0< T\leq 1$. For the functional $D[\rho]$ in \eqref{eq drho} we have the following result.
\begin{proposition}\label{pro fpro2}
We have
\[
D[\rho]\leq 2\sqrt{1-T}.
\]
\end{proposition}  
\begin{proof}
Let us first prove that, if $f,g\in L^2(\mathbb{R})$, with $\|f\|_2=\|g\|_2=1$ we have 
\begin{equation}\label{eq fdist}
\| |f\rangle\langle f|- |g\rangle\langle g| \|_{S_1}=2\sqrt{1-|\langle f,g\rangle|^2}.
\end{equation}
We can restrict our analysis to a plane containing $f$ and $g$. Let $\tilde{f}$ be a function in such a plane, orthogonal to $f$ and with $\|\tilde{f}\|_2=1$. The self-adjoint operator $|f\rangle\langle f|- |g\rangle\langle g|$ is represented, with respect to the orthonormal basis $f,\tilde{f}$, by the matrix 
\[
\begin{pmatrix}
1-|\langle g,f\rangle|^2 & -\langle g,f\rangle \langle \tilde{f},g\rangle\\
-\langle g,\tilde{f}\rangle\langle f,g\rangle&-|\langle g,\tilde{f}\rangle|^2
\end{pmatrix},
\]
that has eigenvalues $\pm|\langle g,\tilde{f}\rangle |$. Hence
\[
\| |f\rangle\langle f|- |g\rangle\langle g| \|_{S_1}=2|\langle g,\tilde{f}\rangle |=2\sqrt{1-|\langle f,g\rangle|^2},
\]
which proves \eqref{eq fdist}. 

Now, for a density matrix $\rho$ as in \eqref{eq fdec} we have, for every $z\in\mathbb{R}^2$,
\begin{align*}
||\rho- |\varphi_z\rangle\langle \varphi_z|\|_{S_1}&=\|\sum_j p_j (|f_j\rangle \langle f_j|-|\varphi_z\rangle\langle \varphi_z|)\|_{S_1}\\
&\leq \sum_j p_j \||f_j\rangle \langle f_j|-|\varphi_z\rangle\langle \varphi_z|\|_{S_1}\\
&\leq 2\sum_j p_j \sqrt{1-|\langle f_j,\varphi_z\rangle|^2}\\
&\leq 2\sqrt{1-\sum_j p_j |\langle f_j,\varphi_z\rangle|^2}\\
&=2\sqrt{1-u_\rho(z)},
\end{align*}
where we used \eqref{eq fdist} and the fact that the function $x\mapsto\sqrt{1-x}$ is concave.  

Taking the infimum over $z\in\mathbb{R}^2$ gives the desired result. 
\end{proof}

\begin{remark}\label{rem pure} In the case of pure states, i.e. $\rho=|f\rangle \langle f|$ where $f\in L^2(\mathbb{R})$, with $\|f\|_2=1$, we have $T=\max_{z\in\mathbb{R}^2} |\mathcal{V}f(z)|^2$, and the inequalities in the proof of Proposition \ref{pro fpro2} are in fact equalities. Hence 
\[
D[\rho]=2\sqrt{1-T}.
\]
Denoting, with abuse of notation,  
\[
D[f]:= \inf_{(x_0,\omega_0)\in\mathbb{R}^2,\, |c|=1} \|f-c\varphi_{(x_0,\omega_0)}\|_2,
\]
it is easy to check that (cf. \cite[Lemma 2.5]{GGRT})
\[
D[f]= \sqrt{2(1-\sqrt{T})}.
\]
Hence 
\[
\frac{1}{2}D[\rho]\leq D[f]\leq \frac{1}{\sqrt{2}}D[\rho].
\]
In this connection we also note that
$$
D[\rho] = 2 \inf_{(x_0,\omega_0)\in\mathbb{R}^2,\, c\in\C} \|f-c\varphi_{(x_0,\omega_0)}\|_2 \,,
$$
where we note that the optimization over all $c\in\C$ without a constraint on their absolute value. This formula follows easily by carrying out the optimization over $c$ explicitly and comparing with the expression \eqref{eq fdist} for $D[\rho]$.
\end{remark}

We now show that the covariant symbol $u_\rho$ of a density matrix $\rho$ enjoys a log-subharmonic type property that will be crucial in the following. 
\begin{proposition}\label{pro logsub}
Let $u_\rho$ be the covariant symbol of a density matrix $\rho$, as in \eqref{eq fcov}. We have 
\[
\Delta \ln u_\rho(z)\geq -4\pi
\] 
for all $z\in\mathbb{R}^2$ such that $u_\rho(z)>0$.
\end{proposition}
\begin{proof}
We use again the representation \eqref{eq ffock}. Setting $G_j(z)=\sqrt{p_j} F_j(z)$ we have to prove that \[
\partial \bar{\partial} \ln \sum_j |G_j|^2\geq 0
\]
where the above sum is strictly positive. Indeed, we have  
\begin{align*}
\partial \bar{\partial} \ln \sum_j |G_j|^2&=\partial \frac{\sum_j G_j \bar{\partial} \bar{G}_j}{\sum_j|G_j|^2}\\
&=\frac{\sum_j |G_j|^2 \sum_k|\partial G_k|^2-| \sum_j \bar{G}_j \partial G_j|^2}{(\sum_j |G_j|^2)^2},
\end{align*}
which is nonnegative by the Cauchy--Schwarz inequality. 
\end{proof}

The following proposition provides information about the distribution function of the covariant symbol of a density matrix. The properties below are a consequence of Lemma \ref{lem pre}, Proposition \ref{pro fpro2} and Proposition \ref{pro logsub}, combined with a general machinery developed in the series of paper \cite{NiTi, Fr,KuNiOCTi}. 
\begin{theorem}\label{lem1}
	Let $\rho$ be a density matrix and $u_\rho$ its covariant symbol, as in \eqref{eq fcov}. Let
	\[
	\mu(t):=|\{u_\rho>t\}|
	\qquad\text{for all}\ t>0.
	\]
	\begin{itemize}
		\item[a)] The function $\mu(t)$ is nonincreasing and absolutely continuous on the compact subsets of $(0,+\infty)$.\medskip
		\item[b)] Setting $T=\max_{z\in\mathbb{R}^2} u_\rho(z)$ as in \eqref{eq fT}, we have $T\in (0,1]$, and $T=1$ if and only if $\rho=|\varphi_{z_0}\rangle \langle \varphi_{z_0} |$ for some $z_0\in\mathbb{R}^2$. \medskip
		\item[c)] $\mu'(t)\leq -1/t$ for almost every $t\in(0,T)$, and $\mu(t)=0$ for $t\geq T$.\medskip
		\item[d)] If $\rho=|\varphi_{z_0}\rangle \langle \varphi_{z_0} |$ for some $z_0\in\mathbb{R}^2$ (i.e. $T=1$), $|\{u_\rho>t\}|=(-\ln t)_+=:\mu_0(t)$. \par\noindent
		If $\rho\not=|\varphi_{z_0}\rangle \langle \varphi_{z_0} |$ for every $z_0\in\mathbb{R}^2$ (i.e. $T<1$) there is a unique ${t^\ast}\in (0,T)$ such that $\mu({t^\ast})=\mu_0({t^\ast})$, and $\mu(t)> \mu_0(t)$ for $0<t<t^\ast$ and  $\mu(t)< \mu_0(t)$  for $t^\ast<t<1$.
	\end{itemize} 
\end{theorem}
\begin{proof} (a) Lemma \ref{lem pre} guarantees that the set of critical points of $u$ has measure zero, so that the desired result follows from the coarea formula; in particular
\[
\mu'(t)=-\int_{\{u_\rho=t\}}\frac{1}{|\nabla u_\rho|}d\mathcal{H}^1\qquad a.e.\ t\in (0,\max u_\rho),
\]
where $\mathcal{H}^1$ is the one dimensional Hausdorff measure; see e.g. the proof of \cite[Lemma 3.2]{NiTi}). \par
(b) If $T=1$ we have $\rho=|\varphi_{z_0}\rangle \langle \varphi_{z_0} |$ for some $z_0\in\mathbb{R}^2$, as a consequence of Proposition \ref{pro fpro2}, since the set of rank one projectors $\{|\varphi_z\rangle\langle  \varphi_z|:\ z\in\mathbb{R}^2\}$ is easily seen to be closed in the space of trace class operators. Viceversa, if $\rho=|\varphi_{z_0}\rangle \langle \varphi_{z_0} |$ for some $z_0\in\mathbb{R}^2$, we have $u_\rho(x,\omega)=|\mathcal{V}\phi(x,\omega)|^2=e^{-\pi(x^2+\omega^2)}$ and therefore $T=1$. 
\par
(c) The differential inequality is a direct consequence of Lemma \ref{lem pre} and the log-suharmonic property in Proposition \ref{pro logsub} in view of a general argument involving the coarea formula and the isoperimetric inequality (see the proofs of \cite[Proposition 3.3]{NiTi}, \cite[Theorem 2.1]{Ku} and \cite[Lemma 12]{Fr}). We sketch this computation for the benefit of the reader: for a.e. $t\in (0,T)$ we have, 
\begin{align}\label{eq iso}
4\pi\mu(t)&\leq \mathcal{H}^1(\{u_\rho=t\})^2\leq \int_{\{u_\rho=t\}}\frac{1}{|\nabla u_\rho|}d\mathcal{H}^1 \int_{\{u_\rho=t\}}|\nabla u_\rho|d\mathcal{H}^1\\
&=-\mu'(t)t\int_{\{u_\rho=t\}}\frac{|\nabla u_\rho|}{u_\rho}d\mathcal{H}^1\nonumber \\
&\leq -\mu'(t)t\iint_{\{u_\rho>t\}}-\Delta \ln u_\rho(x,y)\, dx\,dy\nonumber\\
&\leq -4\pi t \mu'(t) \mu(t). \nonumber
\end{align}

Hence $\mu'(t)\leq -1/t$. Alternatively, one can rely on the general result \cite[Theorem 1]{KuNiOCTi} (that was proved following the same argument), applied to the function $\ln u$ and with $F(t)=e^t$. \par
(d) The first statement is clear by the computation in the proof of part (b). The second statement is similar but slightly stronger than the statement of \cite[Lemma 3.1]{KuNiOCTi}. The function $\mu_0(t)-\mu(t)$ is continuous and nondecreasing on $(0,T)$,  and 
 \begin{equation}\label{eq int}
	\int_0^{+\infty} \mu_0(t)\, dt=\int_0^{+\infty} \mu(t)\, dt=\iint_{\mathbb{R}^2} u_\rho(x, \omega)\,dx d\omega=1,
\end{equation}
where the last equality follows from \eqref{eq fcov} and the fact that $\mathcal{V}:L^2(\mathbb{R})\to L^2(\mathbb{R}^2)$ is an isometry. Hence there is at least one point ${t^\ast}\in (0,T)$ such that $\mu({t^\ast})=\mu_0({t^\ast})$. It remains to prove its uniqueness. We argue as in \cite[Remark 3.5]{NiTi} and the proof of \cite[Corollary 2.2]{GGRT}. Namely, suppose that $\mu(t)=\mu_0(t)$ in some open interval $I\subset (0,T)$ of positive length. Then $\mu'(t)=-1/t$ fo every $t\in I$ and we have equalities in the chain of inequalities \eqref{eq iso}. Hence, for at least one $t\in (0,T)$, the super-level set $\{u_\rho>t\}$ is a ball, say $B$, and $\Delta \ln u_\rho=-4\pi$ in $B$. Without loss of generality we can suppose that $B$ is centered at the origin. Since $u_\rho$ is constant ($=t$) on the boundary $\partial B$, by the uniqueness of the solution of the Dirichlet problem we see that $\log u_\rho(z)=-\pi |z|^2+c$ in $B$ for some $c\in\mathbb{R}$. Hence $u_\rho(z)=e^c e^{-\pi |z|^2}$ for $z\in B$ and therefore for every $z\in\mathbb{R}^2$ by real-analyticity (cf. Lemma \ref{lem pre}). Since $e^c=T=1$ by \eqref{eq int}, this is a contradiction.  
\end{proof}

The main ingredient in the proof of \eqref{SF} is the function  
\begin{equation}
	\label{defH}
	H(t):=\int_0^t (\mu(\tau)-\mu_0(\tau))\,d\tau \,, \quad t\in [0,1] \,,
\end{equation}
where we used the notation of Theorem \ref{lem1}. Note that $H\in C^0([0,1])\cap C^1((0,1])$, $H(0)=H(1)=0$, $H$ is increasing on $[0,{t^\ast}]$ and decreasing on $[{t^\ast},1]$, so $H\geq 0$ achieves its maximum at ${t^\ast}$. Moreover, since $H''=\mu'-\mu_0'$,
$H$ is concave on $[0,T]$ (by the differential inequality $\mu'\leq \mu_0'$, which is valid almost everywhere on $(0,T)$), and 
convex on $[T,1]$ (since, there, $H''=-\mu_0' > 0$). 

Observe that, if $t\geq T$, then by \eqref{eq int}
\begin{equation}\label{eq agg}
	H(t)=\int_t^1(\mu_0(\tau)-\mu(\tau))\, d\tau=\int_t^1 (-\ln \tau)\, d\tau\geq \frac{(1-t)^2}{2} \,, \qquad t\in[T,1] \,,
\end{equation}
where we used the fact that  $-\ln t\geq 1-t$ for $t>0$. This also implies 
\begin{equation}\label{eq debole}
H(t^\ast)\geq \frac{(1-T)^2}{2}.
\end{equation}
In fact, we will need a stronger lower bound, in terms of $1-T$, rather than $(1-T)^2$. It will be a consequence of the next crucial lemma, where $u_\rho$, $T$ and $\mu_0$
are as in the previous theorem.

\begin{lemma}\label{lemma Paolo} 
For every $t_0\in (0,1)$, there exist a threshold $T_0\in (t_0,1)$
and a constant $C_0>0$, both depending only on $t_0$,
 with the following property:  if $T\geq T_0$, then
\begin{equation}
\label{estmuimproved}
\mu(t)\leq (1+C_0(1-T))\log \frac T t,\quad \forall t\in [t_0,T].
\end{equation}
\end{lemma}
This is a nontrivial generalization of Lemma 2.1 in \cite{GGRT}, which
dealt with the particular case of pure states where $u_\rho(z)=|F(z)|^2 e^{-\pi |z|^2}$
(cf. \eqref{eq ffock}). Since only the 
first two steps of the proof in \cite{GGRT}
need changes,  we will limit the proof to the initial part,
which is different, and refer to
\cite{GGRT} for the last two steps which require no changes. 

\begin{proof}
\noindent{\textsc{Step I. }}    We may assume that $u_\rho(z)$ achieves its maximum $T$ at $z=0$ so that, letting $z=0$ in \eqref{eq ffock} and \eqref{normaliz}, 
and recalling that $\sum_j p_j=1$, one has
    \begin{equation}\label{expT}
      \sum_{j\geq 0} p_j |\aj_0|^2=T, \quad\text{and}\quad 
      \sum_{j\geq 0} p_j \bigl(1-|\aj_0|^2\bigr)=1-T.
    \end{equation}
Using Cauchy--Schwarz and \eqref{normaliz} we can estimate the sum
\begin{equation}
\label{eq101}
\begin{aligned}
  &\left| \sum_{n\geq 2} \aj_n \frac {\pi^{n/2} z^n}{\sqrt{n!}}\right|^2
  \leq \left(\sum_{n\geq 2} |\aj_n|^2\right)
  \left(\sum_{n\geq 2} \frac {\pi^{n} |z|^{2n}}{n!}\right)=\\
\leq &    \bigr(1-|\aj_0|^2\bigr) \bigl(e^{\pi  |z|^2} -1-\pi|z|^2\bigr),
\end{aligned}
\end{equation}
while  the same sum starting at $n=1$ can be estimated  as    
 \begin{equation}
   \label{estRj}
   \left| \sum_{n\geq 1} \aj_n \frac {\pi^{n/2} z^n}{\sqrt{n!}}\right|^2
\leq    \bigl(1-|\aj_0|^2\bigr) \bigl(e^{\pi  |z|^2} -1).
 \end{equation}
From \eqref{normaliz}  we have
\begin{equation}
\label{mqFj}
    |F_j(z)|^2= |\aj_0|^2+
\left| \sum_{n\geq 1} \aj_n \frac {\pi^{n/2} z^n}{\sqrt{n!}}\right|^2
    +2 \Re \left(\overline{\aj_0} \sum_{n\geq 1} \aj_n \frac {\pi^{n/2} z^n}{\sqrt{n!}}\right),
\end{equation}   
hence using \eqref{estRj}, multiplying by $p_j$, and then summing
over $j$, 
one obtains from \eqref{eq ffock}
    \begin{equation}
      \label{expu2}
   \frac{e^{\pi |z|^2} u_\rho(z)}T
   \leq \sum_{j}
   p_j \frac{|\aj_0|^2}{T}+
   \sum_j p_j \frac{\bigl(1-|\aj_0|^2\bigr)\bigl(e^{\pi  |z|^2} -1)}T+
   h(z),
          \end{equation}
where $h(z)$ denotes the harmonic function defined by
    \begin{equation}
    \label{defh2}
    h(z):=
    \frac  2 T\sum_{j\geq 1} p_j \Re \left(\overline{\aj_0} \, \sum_{n\geq 1} \aj_n \frac {\pi^{n/2} z^n}{\sqrt{n!}}\right).
    \end{equation}
Finally, defining the number
\begin{equation}
    \label{defdelta}
\delta:=\sqrt{\frac{1-T}{T}},    
    \end{equation}
recalling \eqref{expT}, and renaming $z=r e^{i\theta}$, we can rewrite \eqref{expu2} as 
\begin{equation}
\label{estu}
\frac {e^{\pi r^2}u_\rho\bigl(re^{i\theta}\bigr)}T\leq 1+\delta^2\bigl(e^{\pi r^2}-1\bigr)+
h\bigl(re^{i\theta}\bigr).
\end{equation}
       
\noindent{\textsc{Step II: }} Estimates for $h$.
Since $z=0$ is a critical point for $u_\rho(z)$, by \eqref{eq ffock} it is also a critical point for
 the function $\sum_j p_j |F_j(z)|^2$, and since in the right hand side of \eqref{mqFj}
 the  coefficient of the linear term in $z$ is
 $2\Re \overline{a}_0^{(j)} \, \aj_1 \pi^{1/2}$,
we must have
    \[
    2\sum_{j\geq 1} p_j \Re \overline{\aj}_0 \, \aj_1 \pi^{1/2}   =0.
    \]
Thus,  \eqref{defh2}
    simplifies to
    \begin{equation}
    \label{defh1}
    h(z)=
    \frac 2 T\sum_{j\geq 1} p_j \Re \left(\overline{\aj_0} \, \sum_{n\geq 2} \aj_n \frac {\pi^{n/2} z^n}{\sqrt{n!}}\right)
    \end{equation}
from which, 
using Cauchy--Schwarz and then  \eqref{eq101}, we obtain 
\begin{equation}
\label{esthq}
\begin{aligned}
  |h(z)|^2 
   &\leq \frac{4}{T^2}
   \left(
  \sum_{j\geq 0} p_j
   |\aj_0|^2\right)\left(\sum_{j\geq 0} p_j  \left|\sum_{n\geq 2} \aj_n \frac {\pi^{n/2} z^n}{\sqrt{n!}}\right|^2\right)\\
   &\leq \frac{4}{T^2} \left(
  \sum_{j\geq 0} p_j
   |\aj_0|^2\right)   \sum_{j\geq 0} p_j \bigr(1-|\aj_0|^2\bigr)\bigl(e^{\pi  |z|^2} -1-\pi|z|^2\bigr)\\
    &=  4\delta^2\bigl(e^{\pi  |z|^2} -1-\pi|z|^2\bigr),
   \end{aligned}
   \end{equation}
having used \eqref{expT} and \eqref{defdelta} in the last passage.
Finally, using the elementary inequality $2(e^x-1-x)\leq x^2 e^x$ with $x=\pi |z|^2$,
we obtain from the last estimate
\begin{equation}
\label{esth}
|h(z)|\leq \sqrt 2\pi\delta |z|^2 e^{\frac{\pi |z|^2}2},\quad\forall z\in\C.
\end{equation}
We now want to estimate, uniformly with respect to the angular variable $\theta$,
the first and second radial derivatives of $h(r e^{i\theta})$. Letting
$z=r e^{i\theta}$ in \eqref{defh1} and differentiating with respect to $r$, one has
\begin{equation}
\label{esthrq}
\left| \frac{\partial h(r e^{i\theta})}{\partial r}\right|
\leq 
 \frac 2 T\sum_{j\geq 1} p_j |\aj_0|  \sum_{n\geq 2} |\aj_n|
 \frac {\pi^{n/2} n r^{n-1}}{\sqrt{n!}}.
\end{equation}
Observe that, since $\frac{n^2}{n!}\leq\frac{2}{(n-2)!}$ for every $n\geq 2$,
\begin{equation}
\label{eq104}
\sum_{n\geq 2} \frac{n^2 \pi^n r^{2(n-1)}}{n!}
\leq 2 \sum_{n\geq 2} \frac{ \pi^n r^{2(n-1)}}{(n-2)!}=2\pi^2 r^2 e^{\pi r^2}
\end{equation}
so that,
similarly to \eqref{eq101}, we now have
\begin{equation*}
\begin{aligned}
  &\left| \sum_{n\geq 2} |\aj_n| \frac {\pi^{n/2} n r^{n-1}}{\sqrt{n!}}\right|^2
  \leq \left(\sum_{n\geq 2} |\aj_n|^2\right)
  \left(\sum_{n\geq 2} \frac {\pi^{n} n^2 r^{2(n-1)}}{n!}\right)
&\leq    \bigr(1-|\aj_0|^2\bigr)2\pi^2 r^2 e^{\pi r^2}.
\end{aligned}
\end{equation*}
Thus, after squaring both sides of \eqref{esthrq}, one can proceed  as done
in \eqref{esthq}, now using the last estimate in place of \eqref{eq101}, to obtain
\begin{equation*}
\left| \frac{\partial h(r e^{i\theta})}{\partial r}\right|^2 \leq 
   \frac{4}{T^2} \left(
  \sum_{j\geq 0} p_j
   |\aj_0|^2\right)   \sum_{j\geq 0} p_j \bigr(1-|\aj_0|^2\bigr)
   2\pi^2 r^2 e^{\pi r^2}
    =  8\delta^2 \pi^2 r^2 e^{\pi r^2},
   \end{equation*}
which for future reference we record as
\begin{equation}
\label{esthr}
        \left| \frac{\partial h(r e^{i\theta})}{\partial r}\right|
        \leq 2\sqrt{2}\delta\pi r e^{\frac{\pi r^2}2},\quad\forall r\geq 0.
\end{equation}
Finally, by a similar argument, now relying on the estimate
\[
\sum_{n\geq 2} \frac{n^2(n-1)^2 \pi^n r^{2(n-2)}}{n!}
\leq C \bigl(1+r^2\bigr)^2 e^{\pi r^2}
\]
in place of \eqref{eq104}, one obtains that
\begin{equation}
\label{esthrrq}
        \left| \frac{\partial^2 h(r e^{i\theta})}{\partial r^2}\right|
        \leq C \delta(1+ r^2) e^{\frac{\pi r^2}2},\quad\forall r\geq 0
\end{equation}
for some universal constant $C>0$.

\noindent{\textsc{Step III. }} Now the proof can be completed exactly as in
steps \textsc{III} and \textsc{IV} of \cite[Lemma 2.1]{GGRT}, with the same notation and symbols except for $u$ in place of $u_\rho$. Indeed,  given $t>0$, defining for fixed $\theta\in[0,2\pi]$  
the function
\begin{equation}
\label{defgt}
g_\theta(r,\sigma):=e^{\pi r^2}\left(\frac t T-\delta^2\right)+\delta^2
-\sigma h\bigl(re^{i\theta}\bigr),\quad
r\geq 0,\quad \sigma\in [0,1]
\end{equation}
as in (2.20) of   \cite{GGRT}, using \eqref{estu} it is immediate to check the validity
of the implication
\begin{equation}
\label{implication}
u_\rho\bigl(re^{i\theta}\bigr)>t\quad \Longrightarrow\quad
 g_\theta(r,1)<1,
\end{equation}
which coincides with (2.19) of \cite{GGRT}. Moreover, although the
harmonic function $h(z)$ is not the same, the relevant estimates
\eqref{esth}, \eqref{esthr} and \eqref{esthrrq} coincide verbatim with
(2.14), (2.17) and (2.18) of \cite{GGRT} and these, together with
\eqref{defgt} and \eqref{implication}, are enough to prove \eqref{estmuimproved} 
as explained in steps 
  \textsc{III} and \textsc{IV} of \cite[Lemma 2.1]{GGRT}. Therefore we do not
  reproduce those steps here, and refer to \cite{GGRT}.
\end{proof}
We can then prove the desired improvement of \eqref{eq debole}. 
\begin{lemma}\label{lemmaggrt}

With the notation of Theorem \ref{lem1} and \eqref{defH}, there is a constant $C>0$ such that for all density matrices $\rho$ one has
	\begin{equation}\label{HT}
		1-T\leq C H({t^\ast}). 
	\end{equation}
\end{lemma}

\begin{proof}
	Using Lemma \ref{lemma Paolo}, together with the properties in Theorem \ref{lem1} and \eqref{eq debole} one can repeat verbatim the proofs of \cite[Corollary 2.2,Lemma 2.4]{GGRT} 	
	 and obtain that, for some constant $C>0$,
	\[
	1-T\leq C\int_0^{s^\ast} (e^{-s}-u^\ast(s))\, ds,
	\]
	where $u^\ast(s)$ is the inverse function of $s=\mu(t)$ and $t^\ast=e^{-s^\ast}$, hence $s^\ast=u(t^\ast)$.  Since 
	\[
	\int_{t^\ast}^1 (\mu_0(t)-\mu(t))\, dt= \int_0^{s^\ast} (e^{-s}-u^\ast(s))\, ds,
	\]
	we obtain \eqref{HT} with the same constant $C$. 
\end{proof}

We are now ready to prove Theorem \ref{thm1}.

\begin{proof}[Proof of Theorem \ref{thm1}]
	In follows from the proof of \cite[Theorem 1.3]{KuNiOCTi} that the integrals in \eqref{SF} could be $-\infty$, but not $+\infty$ and that, if the integral on the right-hand side is $-\infty$, then the same holds for the integral on the left-hand side. Hence we can suppose that both integrals are finite. 
	
	With the notation of Theorem \ref{lem1}, 
	in terms of distribution functions, we have to prove that, if
	\begin{equation}
		\label{A}
		\varepsilon \geq \int_0^1 \Phi'(t)(\mu_0(t)-\mu(t))\,dt,
	\end{equation}  
	then
	\begin{equation}
		\label{B}
		1-T \leq C  \varepsilon,
	\end{equation}  
	for some constant $C$ that depends only on $\Phi$.
	Indeed, \eqref{SF} follows immediately from \eqref{B} and Proposition \ref{pro fpro2}.
	Of course we can suppose $T<1$. 
	
	Let $H(t)$ be the function in \eqref{defH}. The proof of \eqref{B} is elementary and is based on the following steps: 1) prove that
	$H(t_1)\leq C_1\eps$ at some point $t_1$ not too close to $t=0$ nor to $t=1$; 2) prove that this implies $H({t^\ast})\leq C_2\eps$, and conclude using \eqref{HT}.

	\bigskip
	
	\noindent{\textsc{Step I. }} Since $\Phi$ is convex, but not linear, there are two points $0<a<b<1$ such that
	\[
	\Phi'(b)-\Phi'(a)>0 \,.
	\]
	(Here and throughout, for definiteness, by $\Phi'(t)$ we mean e.g.\ the right derivative $\Phi'(t^+)$.) Let $t^\ast$ be as in Theorem \ref{lem1}. Then from \eqref{A}, since $\mu_0(t)-\mu(t)=-H'(t)$,
	\begin{align}
		\label{esti}
		\varepsilon\geq \int_0^1 \Phi'(t)\bigl(-H'(t) \bigr)\,dt=
		\int_0^1 \bigl(\Phi'(t)-\Phi'({t^\ast})\bigr)\bigl(-H'(t) \bigr)\,dt \,.
	\end{align}
	Note that $\bigl(\Phi'(t)-\Phi'({t^\ast})\bigr)\bigl(-H'(t) \bigr)\geq 0$ everywhere.
	Thus, if $a < {t^\ast} < b$, we have
	\begin{align*}
		\varepsilon &\geq 
		\int_0^a \bigl(\Phi'(t)-\Phi'({t^\ast})\bigr)\bigl(-H'(t) \bigr)\,dt
		+
		\int_b^1 \bigl(\Phi'(t)-\Phi'({t^\ast})\bigr)\bigl(-H'(t) \bigr)\,dt
		\\
		&=\int_0^a \bigl(\Phi'({t^\ast})-\Phi'(t)\bigr)H'(t)\,dt
		+
		\int_b^1 \bigl(\Phi'(t)-\Phi'({t^\ast})\bigr)\bigl(-H'(t) \bigr)\,dt\\
		&\geq
		\int_0^a \bigl(\Phi'({t^\ast})-\Phi'(a)\bigr)H'(t)\,dt
		+
		\int_b^1 \bigl(\Phi'(b)-\Phi'({t^\ast})\bigr)\bigl(-H'(t) \bigr)\,dt\\
		&=
		\bigl(\Phi'({t^\ast})-\Phi'(a)\bigr)H(a)
		+
		\bigl(\Phi'(b)-\Phi'({t^\ast})\bigr)H(b).
	\end{align*}
	Letting $t_1:=a$ or $t_1:=b$ in such a way that
	$H(t_1)=\min \{H(a),H(b)\}$, we get
	\begin{equation}
		\label{eq7}
		\eps
		\geq 
		\bigl(\Phi'(b)-\Phi'(a)\bigr)H(t_1).
	\end{equation}
	If, on the other hand, ${t^\ast}\geq b$, then from \eqref{esti}
	\begin{align*}
		\varepsilon &\geq 
		\int_0^a \bigl(\Phi'(t)-\Phi'({t^\ast})\bigr)\bigl(-H'(t) \bigr)\,dt
		=\int_0^a \bigl(\Phi'({t^\ast})-\Phi'(t)\bigr)H'(t)\,dt\\
		&\geq
		\int_0^a \bigl(\Phi'(b)-\Phi'(a)\bigr)H'(t)\,dt
		=
		\bigl(\Phi'(b)-\Phi'(a)\bigr)H(a),
	\end{align*}
	and \eqref{eq7} holds again.
	Similarly, if ${t^\ast}\leq a$, then
	from \eqref{esti}
	\begin{align*}
		\varepsilon &\geq 
		\int_b^1 \bigl(\Phi'(t)-\Phi'({t^\ast})\bigr)\bigl(-H'(t) \bigr)\,dt\\
		&\geq
		\int_b^1 \bigl(\Phi'(b)-\Phi'(a)\bigr)\bigl(-H'(t) \bigr)\,dt
		=
		\bigl(\Phi'(b)-\Phi'(a)\bigr)H(b),
	\end{align*}
	and \eqref{eq7} holds again.
	Summing up, in all cases we get
	\begin{equation}
		\label{estHt1}
		H(t_1)\leq \frac\eps{\Phi'(b)-\Phi'(a)},\quad\text{where either $t_1=a$ or $t_1=b$}.
	\end{equation}

	\bigskip
	
\noindent{\textsc{Step II. }} Now that $t_1\in\{a,b\}$ has been fixed, we may have
	$t_1\leq {t^\ast}$, or ${t^\ast}<t_1\leq T$, or $t_1>T$.

	If $t_1\leq {t^\ast}$,  since $H(0)=0$ and $H$ is concave on
	$[0,T]$, the ratio $H(t)/t$ is nonincreasing on $(0,T]$, so that
	\[
	H({t^\ast})\leq
	\frac {H({t^\ast})}{{t^\ast}}\leq\frac{H(t_1)}{t_1} \leq \frac{H(t_1)}{a}\leq \frac{\eps}{a (\Phi'(b)-\Phi'(a))}
	\]
	by \eqref{estHt1}. Combining with \eqref{HT}, we find
	\begin{equation*}
		1-T\leq \frac{C\eps}{a (\Phi'(b)-\Phi'(a))}.
	\end{equation*}

	If ${t^\ast}\leq t_1< T$, we consider the function $\tilde{H}(t)$ given by 
	\[
	\tilde{H}(t) :=
	\begin{cases}
		H(t) & \text{if}\ t\in [0,T] \,, \\
		H(T)+H'(T)(t-T) & \text{if}\ t\in (T,t_2] \,,
	\end{cases}
	\]
	where $H(T) =\int_T^1 (-\ln \tau)\, d\tau$ and $t_2:=(T-1)/\ln T\in (T,1)$ is the point $t$ where  $H(T)+H'(T)(t-T)=0$. Hence $\tilde{H}$ is concave on $[0,t_2]$ and $\tilde{H}(t_2)=0$. Since $t_2\to 1$ as $T\to 1$, we see that there is a $c=c(b)>0$ such that if $1-T<c$ we have, say, $t_2-b>\frac{1-b}{2}$. 
	
	Arguing as above, and assuming for the moment $1-T<c$, we have 
	\[
	H({t^\ast})\leq
	\frac {H({t^\ast})}{t_2-{t^\ast}}\leq\frac{H(t_1)}{t_2-t_1} \leq \frac{H(t_1)}{t_2-b}\leq \frac{2\eps}{(1-b) (\Phi'(b)-\Phi'(a))}.
	\]
	Combining with \eqref{HT}, we find
	\begin{equation*}
		1-T\leq \frac{2C\eps}{(1-b)(\Phi'(b)-\Phi'(a))}.
	\end{equation*}
	If $1-T\geq c$ we simply use the fact that, by \eqref{eq agg} (applied with $t=T$), 
	\[
	H(t_1)\geq H(T)\geq \frac{(1-T)^2}{2}\geq \frac c2\, (1-T)
	\]
	and one concludes using \eqref{estHt1}. 
	
	Finally, if $t_1> T$, then from  \eqref{estHt1} and \eqref{eq agg} (applied with $t=t_1$)  we get
	\[
	\frac\eps{\Phi'(b)-\Phi'(a)}\geq H(t_1)\geq \frac{(1-t_1)^2}2
	\geq \frac{(1-b)^2}2 \,,
	\]
	which is a lower bound for $\eps$ (so that \eqref{B} is trivial). In particular, we have
	\[
	\frac\eps{\Phi'(b)-\Phi'(a)}
	\geq \frac{(1-b)^2}2 \geq \frac{(1-b)^2}2 (1-T)
	\]
	and thus also in this case
	\begin{equation}
		\label{est2}
		1-T\leq\frac{2\eps}{(1-b)^2 (\Phi'(b)-\Phi'(a))}.
	\end{equation}
	This completes our first proof of Theorem \ref{thm1}.
\end{proof}
\begin{remark}\label{multidim} Theorem \ref{thm1} holds in dimension $d$ essentially with the same statement -- just replace $\R^2$ by $\R^{2d}$ in the integrals. Precisely, consider the $L^2(\R^d)$-normalized Gaussian
$$
\phi(x) := 2^{d/4} e^{-\pi |x|^2} \,,
\qquad\text{for all}\ x\in\R^d \,,
$$
as well as the coherent states, for $(x_0,\omega_0)\in\R^d\times \R^d$,
$$
\phi_{(x_0,\omega_0)}(x) = e^{2\pi i\omega_0\cdot x} \phi(x-x_0) 
\qquad\text{for all}\ x\in\R^d,
$$
and the corresponding transform 
$$
(\mathcal V f)(x_0,\omega_0) = \int_{\R^d} \overline{\phi_{(x_0,\omega_0)}(x)} f(x)\,dx.
$$
Also, if $\rho$ is a density matrix on $L^2(\mathbb{R}^d)$, we set
$$
D[\rho] := \inf_{x_0, \omega_0\in\R^d} \|\rho- |\phi_{(x_0,\omega_0)}\rangle \langle \phi_{(x_0,\omega_0)}| \|_{S_1}.
$$
Setting $u_\rho(x,\omega)=\langle \varphi_{(x,\omega)}, \rho \varphi_{(x,\omega)}\rangle$ as in \eqref{eq fcov}, Lemma \ref{lem pre} easily extends in dimension $d$. Proposition \ref{pro fpro2} and its proof continue to hold verbatim, whereas in Proposition \ref{pro logsub} the estimate now reads $\Delta \ln u_\rho\geq -4\pi d$. The distribution function $\mu(t):=|\{u_\rho>t\}|$ enjoys the same properties as in Theorem \ref{lem1}, except that the inequality $\mu'(t)\leq -1/t$ is now replaced by
\[
\mu'(t)\leq -\frac{d\mu(t)^{1-1/d}}{(d!)^{1/d} t} \quad {\rm for\ a.e.}\ t\in (0,T),
\]
and
\[
\mu_0(t):=|\{|\mathcal{V} \varphi|^2>t\}|
=\frac{1}{d!}\Big(\ln_{+}\frac{1}{t} \Big)^d.
\]
Lemma \ref{lemma Paolo} extends as well (the computations in the proof are a little more involved; cf. also \cite[(7.25)]{GGRT}) and Lemma \ref{lemmaggrt} continues to hold as it is (the function $e^{-s}$ in the proof has to be replaced by $e^{-(d! s)^{1/d}}$, that is the inverse function of $\mu_0(t)$, $t\in (0,1]$, in this case). 

Also, formula \eqref{eq agg} now reads

\[
	H(t)=\int_t^1(\mu_0(\tau)-\mu(\tau))\, d\tau=\int_t^1 \mu_0(\tau)\, d\tau\geq \frac{(1-t)^{d+1}}{(d+1)!}\qquad t\in[T,1].
\]
Then the generalization to higher dimensions of the proof of Theorem \ref{thm1} requires only obvious changes.
\end{remark}


\section{A second proof of the main result}
We give a second proof of Theorem \ref{thm1} in the case of pure states, namely 
\[
\rho=|f\rangle\langle f |\qquad{\rm with}\ \|f\|_2=1,
\]
and therefore in the statement of Theorem \ref{thm1}
\[
u_\rho(x,\omega)=|\mathcal{V}f(x,\omega)|^2. 
\]
Also, if $\|f\|_2=1$, we set (with abuse of notation),
\[
D[f]:= \inf_{(x_0,\omega_0)\in\mathbb{R}^2,\, |c|=1} \|f-c\varphi_{(x_0,\omega_0)}\|_2,
\]
which is equivalent to the quantity $D[\rho]$ in \eqref{eq drho} if $\rho=|f\rangle\langle f |$; see Remark \ref{rem pure}. \par\medskip 

This proof is related to Hardy--Littlewood majorization theory. While we present the proof in a self-contained way and will not use anything from this theory, it might be helpful for the reader to view it from this point of view (see, e.g., \cite[Theorems 108, 249, 250]{HaLiPo}, \cite[Corollary 2.1]{AlTrLi}, \cite[Theorem 15.27]{Si3} and also \cite[Chapter 2, Proposition 3.3]{BeSh}). Using majorization theory one can show that the generalized Wehrl conjecture and the Faber--Krahn inequality are equivalent. The main idea of this section is to show that stability for the Faber--Krahn inequality implies stability for the generalized Wehrl entropy bound.

\subsection*{Proof of the generalized Wehrl conjecture given the Faber--Krahn inequality}

In \cite[Section 5]{Fr} we have shown how the generalized Wehrl conjecture implies the Faber--Krahn theorem of \cite{NiTi}. As a warm up, let us begin by showing the converse. For the sake of conciseness, we do not characterize the cases of equality.

Let $\Phi:[0,1]\to\R$ be continuous and convex with $\Phi(0)=0$. For simplicity we also assume that $\Phi'(0)>-\infty$. (Here $\Phi'(0)$ denotes the right-sided derivative of $\Phi$, which exists by monotonicity. The assumption that it is $>-\infty$ excludes the important special case $\Phi(u) = u\ln u$, which we will discuss separately afterwards.) We can write, using $\Phi(0)=0$ and $\Phi'(0)>-\infty$,
\begin{equation}
	\label{eq:superposition}
	\Phi(u) = \Phi'(0) u + \int_0^\infty (u-\tau)_+ \Phi''(\tau)\,d\tau
	\qquad\text{for all}\ u\in[0,1] \,.
\end{equation}
Here the expression $\Phi''(\tau)\,d\tau$ should be interpreted as a (not necessarily absolutely continuous) measure on $[0,1]$. It follows that
$$
\iint_{\R^2} \Phi(|\mathcal Vf(x,\omega)|^2)\,dx\,d\omega = \Phi'(0) + \int_0^\infty \iint_{\R^2} \left( |\mathcal Vf(x,\omega)|^2 - \tau \right)_+ dx\,d\omega\, \Phi''(\tau)\,d\tau.
$$
We fix $\tau>0$ and set
$$
E := \{ |\mathcal Vf|^2>\tau \} \,.
$$
Then, by the Faber--Krahn inequality of \cite{NiTi},
\begin{align*}
	\iint_{\R^2} \left( |\mathcal Vf(x,\omega)|^2 - \tau \right)_+ dx\,d\omega
	& = \iint_E |\mathcal Vf(x,\omega)|^2 \,dx\,d\omega - \tau |E| \\
	& \leq \iint_{E^*} |\mathcal V\phi(x,\omega)|^2 \,dx\,d\omega - \tau |E^*| \\
	& = \iint_{E^*} \left( |\mathcal V\phi(x,\omega)|^2 - \tau \right) dx\,d\omega \,,
\end{align*}
where $E^*$ denotes the disk in $\R^2$, centered at the origin, with the same area as $E$. Now it is immediate that
\begin{align*}
	\iint_{E^*} \left( |\mathcal V\phi(x,\omega)|^2 - \tau \right) dx\,d\omega 
	& \leq \iint_{E^*} \left( |\mathcal V\phi(x,\omega)|^2 - \tau \right)_+ dx\,d\omega \\
	& \leq \iint_{\R^2} \left( |\mathcal V\phi(x,\omega)|^2 - \tau \right)_+ dx\,d\omega \,.
\end{align*}
Inserting this into the above expression, we find
\begin{align*}
	\iint_{\R^2} \Phi(|\mathcal Vf(x,\omega)|^2)\,dx\,d\omega & \leq \Phi'(0) + \int_0^\infty \iint_{\R^2} \left( |\mathcal V\phi(x,\omega)|^2 - \tau \right)_+ dx\,d\omega\, \Phi''(\tau)\,d\tau \\
	& = \iint_{\R^2} \Phi(|\mathcal V\phi(x,\omega)|^2)\,dx\,d\omega \,.
\end{align*}
This is the generalized Wehrl entropy inequality of Lieb and Solovej in the special case $\Phi'(0)>-\infty$.

In case $\Phi'(0)=-\infty$, we apply the above inequality with $\Phi(u)$ replaced by the function $\max\{\Phi(u),-(1/\varepsilon) u\}$ and let $\epsilon\to 0$. In the original Wehrl case, i.e.\ $\Phi(u)=u\ln u$, we can alternatively write
\begin{equation}
	\label{eq:superpositionwehrl}
	u\ln u = \int_0^1 \left( (u-\tau)_+ - u \right) \frac{d\tau}{\tau} + \int_1^\infty (u-\tau)_+ \frac{d\tau}{\tau} + u
	\qquad\text{for all}\ u\in [0,\infty) \,.
\end{equation}
This can be verified by explicit computation. Given this formula, the argument is precisely the same as before.


\subsection*{Stability for fixed $\tau$}

To get stability for the generalized Wehrl inequality, we want to follow the same route as above and deduce it from the stability for the Faber--Krahn inequality. The goal is therefore to first prove stability for the special case $\Phi(u)=(u-\tau)_+$ with some fixed $\tau>0$. We shall prove

\begin{theorem}\label{stabtau}
	There is a constant $c>0$ such that for any $\tau>0$ and any $f$ with $\|f\|_2=1$,
	$$
	\iint_{\R^2} (|\mathcal V f|^2-\tau)_+ \,dx\,d\omega \leq \left( 1 -  c \tau D[f]^2 \right) \iint_{\R^2} (|\mathcal V \phi|^2-\tau)_+ \,dx\,d\omega \,.
	$$
	The constant $c$ is explicit and can be chosen as $\min\{ \frac{c_0}2,\frac1{16}\}$, where $c_0$ is the constant in the stability version of the Faber--Krahn inequality. 
\end{theorem}

We recall that the stability version of the Faber--Krahn inequality, due to \cite{GGRT}, reads
$$
\iint_\Omega |\mathcal Vf(x,\omega)|^2\,dx\,d\omega \leq \left(1- c_0 e^{-|\Omega|} D[f]^2 \right) \iint_{\Omega^*} |\mathcal V\phi(x,\omega)|^2\,dx\,d\omega \,.
$$

\begin{proof}
	We go through the above argument, setting again $E:=\{|\mathcal Vf|^2>\tau\}$. By the above stability version of the Faber--Krahn inequality, we obtain
	\begin{align*}
		\iint_{\R^2} \left( |\mathcal Vf(x,\omega)|^2 - \tau \right)_+dx\,d\omega
		& \leq (1 - c_0 e^{-|E|} D[f]^2) \iint_{E^*} |\mathcal V\phi(x,\omega)|^2 \,dx\,d\omega - \tau |E^*| \\
		& = (1 - c_0 e^{-|E|} D[f]^2) \iint_{E^*} \left( |\mathcal V\phi(x,\omega)|^2 - \tau \right) dx\,d\omega \\
		& \quad - c_0 e^{-|E|} D[f]^2 \tau |E| \\
		& \leq (1 - c_0 e^{-|E|} D[f]^2) \iint_{E^*} \left( |\mathcal V\phi(x,\omega)|^2 - \tau \right) dx\,d\omega \,.
	\end{align*}
	Note that, by expanding the square of the $L^2$ norm and optimizing over $c$,
	\begin{equation}
		\label{eq:distancesmall}
		D[f]^2 = 2 \left( 1- \sup_{z_0} |(\phi_{z_0},f)| \right) \leq 2 \,.
	\end{equation}
	Thus, $c_0 e^{-|E|} D[f]^2 \leq 2c_0$ and by assuming that $c_0\leq 1/2$, we may assume that the prefactor $1 - c_0 e^{-|E|} D[f]^2$ is nonnegative. Therefore we can argue as before and use
	\begin{equation}
		\label{eq:bathtub1}
		\iint_{E^*} \left( |\mathcal V\phi(x,\omega)|^2 - \tau \right) dx\,d\omega \leq
	\iint_{\R^2} \left( |\mathcal V\phi(x,\omega)|^2 - \tau \right)_+ dx\,d\omega \,.
	\end{equation}
	Thus, we have shown the bound
	\begin{equation}
		\label{eq:tauinitial}
		\iint_{\R^2} \left( |\mathcal Vf(x,\omega)|^2 - \tau \right)_+dx\,d\omega
		\leq (1 - c_0 e^{-|E|} D[f]^2) \iint_{\R^2} \left( |\mathcal V\phi(x,\omega)|^2 - \tau 	\right)_+ dx\,d\omega \,.
	\end{equation}
	This is almost of the form claimed in Theorem \ref{stabtau}, except that it contains $|E|$, which depends on $f$, and which we still need to bound.
	
	We set
	$$
	F_\tau := \{ |\mathcal V\phi|^2 >\tau \} \,.
	$$
	Using
	$$
	\mathcal V\phi(x,\omega) = e^{i\pi x\omega} e^{-\frac\pi 2(x^2+\omega^2)} \,,
	$$
	we see that $F_\tau$ is a disk with
	$$
	|F_\tau|=\ln\frac1\tau \,.
	$$
	
	We distinguish two cases. The first one is where
	$$
	e^{-|E|} \geq \frac12 e^{-|F_\tau|} = \frac\tau 2 \,.
	$$
	In this case, we can insert this bound into \eqref{eq:tauinitial} and obtain a bound of the form claimed in Theorem \ref{stabtau}.
	
	In the opposite case, we have
	$$
	e^{-|E|} < \frac12 e^{-|F_\tau|} \,,
	$$
	which we can write as
	\begin{equation}
		\label{eq:case2}
		e^{-|F_\tau|} - e^{-|E|} > \frac12 e^{-|F_\tau|}= \frac\tau 2 \,.
	\end{equation}
	In this case we only use the Faber--Krahn inequality, but not its stability version. Our goal is to improve on inequality \eqref{eq:bathtub1}. To do so, we write
	\begin{align}
		\label{eq:bathtubformula}
		\iint_{E^*} \left( |\mathcal V \phi|^2 - \tau \right) dx\,d\omega 
		& = \iint_{\R^2} \left( |\mathcal V \phi|^2 - \tau \right)_+ dx\,d\omega \notag \\
		& - \iint_{F_\tau \setminus E^*} \left( |\mathcal V \phi|^2 - \tau \right)_+ dx\,d\omega
		- \iint_{E^* \setminus F_\tau} \left( |\mathcal V \phi|^2 - \tau \right)_- dx\,d\omega \,.
	\end{align}
	Previously to arrive at \eqref{eq:bathtub1}, we dropped the last two terms. Now we keep them and we will get a remainder bound out of these.
	
	Note that $F_\tau$ and $E^*$ are concentric disks, so at most one of the two extra integrals is nonzero. Our assumption on $e^{-|E|} < \frac12 e^{-|F_\tau|}$ implies that $|E|> |F_\tau|$, so only the last term is nonzero. We define $R\geq 0$ by
	$$
	|E| = \pi R^2
	$$
	and compute explicitly using the above form of $\mathcal V\phi$. We obtain
	\begin{align}
		\label{eq:bathtub}
		\iint_{E^* \setminus F_\tau} \left( |\mathcal V \phi|^2 - \tau \right)_- dx\,d\omega
		& = 2\pi \int_{(\frac1\pi \ln\frac1\tau)^{1/2}}^R ( \tau - e^{-\pi r^2}) r\,dr \notag \\
		& = \int_{\ln\frac1\tau}^{\pi R^2} ( \tau - e^{-\rho})\,d\rho \notag \\
		& = e^{-|E|} - e^{-|F_\tau|} + \tau |E| - \tau |F_\tau| \,.
	\end{align}
	Now we use the elementary inequality
	\begin{equation*}
		s-\ln s - 1 \geq \frac12(1-s)^2
		\qquad\text{for all}\ 0<s\leq 1 \,,
	\end{equation*}
	which follows by writing
	\begin{align*}
		s-\ln s - 1 = \int_s^1 \left( \frac1t- 1\right)dt = \int_s^1 \int_t^1 \frac{du}{u^2}\,dt \geq \int_s^1 \int_t^1 du\,dt = \frac12(1-s)^2 \,.
	\end{align*}
	Taking $s = e^{-|E|} / e^{-|F_\tau|} = \tau e^{-|E|}$, we deduce that
	$$
	e^{-|E|} - e^{-|F_\tau|} + \tau |E| - \tau |F_\tau| \geq \frac1{2\tau} (e^{-|F_\tau|} - e^{-|E|})^2 \,.
	$$
	This, when inserted into \eqref{eq:bathtub}, gives
	\begin{equation}
		\label{eq:bathtubrem}
			\iint_{E^* \setminus F_\tau } \left( |\mathcal V \phi|^2 - \tau \right)_+ dx\,d\omega
		\geq \frac1{2\tau} (e^{-|F_\tau|} - e^{-|E|})^2 \,,
	\end{equation}
	and, returning to \eqref{eq:bathtubformula}, we find
	$$
	\iint_{E^*} \left( |\mathcal V \phi|^2 - \tau \right) dx\,d\omega 
	\leq \iint_{\R^2} \left( |\mathcal V \phi|^2 - \tau \right)_+ dx\,d\omega \\
	- \frac1{2\tau} (e^{-|F_\tau|} - e^{-|E|})^2 \,.
	$$
	This is the desired improvement of \eqref{eq:bathtub1}.
	
	Thus, if we go through the argument at the beginning of this section and use our improved bound, we get
	\begin{align*}
		\iint_{\R^2} (|\mathcal Vf|^2-\tau)_+\,dx\,d\omega 
		& \leq \iint_{E^*} (|\mathcal V\phi|^2-\tau)\, dx\,d\omega \\
		& \leq \iint_{\R^2} (|\mathcal V\phi|^2-\tau)_+ \,dx\,d\omega - \frac1{2\tau} (e^{-|F_\tau|} - e^{-|E|})^2 \,.
	\end{align*}
	According to our assumption \eqref{eq:case2} and the bound \eqref{eq:distancesmall}, we have
	$$
	\frac1{2\tau} (e^{-|F_\tau|} - e^{-|E|})^2 \geq \frac{\tau}{8} \geq \frac{\tau}{16} D[f]^2 \,.
	$$
	In order to write the inequality in a multiplicative form, we compute
	\begin{equation}
		\label{eq:entropytau}
		\iint_{\R^2} (|\mathcal V\phi|^2-\tau)_+ \,dx\,d\omega = 1 - \tau - \tau \ln \frac1\tau \,.
	\end{equation}
	In particular, $\iint_{\R^2} (|\mathcal V\phi|^2-\tau)_+ \,dx\,d\omega\leq 1$ and therefore
	$$
	\frac{\tau}{16} D[f]^2 \geq \frac{\tau}{16} D[f]^2 \iint_{\R^2} (|\mathcal V\phi|^2-\tau)_+ \,dx\,d\omega \,.
	$$
	This gives the desired bound in the second case.
\end{proof}


\subsection*{Stability for general $\Phi$}

With Theorem \ref{stabtau} at hand, it is easy to prove Theorem \ref{thm1} (recall that here $\rho=|f\rangle\langle f|$). We will prove the bound \eqref{SF} for pure states with constant
$$
c_\Phi := c \int_0^1 \tau (1-\tau - \tau \ln\frac1\tau)\, \Phi''(\tau)\,d\tau \,,
$$
where $c$ is the constant from Theorem \ref{stabtau}.

Indeed, when $\Phi'(0)>-\infty$, then this bound follows from the superposition formula \eqref{eq:superposition}, the expression \eqref{eq:entropytau} for the maximum and the stability bound for $\Phi(u)=(u-\tau)_+$ from Theorem \ref{stabtau}.

If $\Phi'(0)=-\infty$ we can apply the same approximation argument as at the beginning of this section. Alternative, in the original Wehrl case we can use the superposition formula \eqref{eq:superpositionwehrl}. In both ways we obtain the bound \eqref{eq:wehrlintro} with constant
$$
c_* := c \int_0^1 (1-\tau - \tau \ln\frac1\tau) \,d\tau = \frac c4 \,,
$$	
where again $c$ is the constant from Theorem \ref{stabtau}.


\section{A logarithmic Sobolev inequality for entire functions}\label{sec log-sobolev}

As before, we denote the two-dimensional Lebesgue measure in $\C$ by $dA(z)$ and introduce, for a parameter $h>0$, a probability measure on $\C$ by
$$
d\mu_h(z) = h^{-1} e^{-\pi |z|^2/h} \,dA(z) \,.
$$
By $\mathcal F_h^2$ we denote the Fock space, that is, the space of all entire functions that are square integrable with respect to $\mu_h$. This is a Hilbert space with respect to the inner product corresponding to the norm
$$
\| F \|_{\mathcal F^2_h} := \left( \int_\C |F(z)|^2\,d\mu_h(z) \right)^\frac12.
$$
The log-Sobolev inequality for entire functions states that
\begin{equation}
	\label{eq:logsob}
	\frac{h}{\pi} \int_\C |\partial_z F|^2\,d\mu_h(z)  \geq \int_\C |F(z)|^2 \ln |F(z)|^2 \,d\mu_h(z)
\end{equation}
for all $F\in\mathcal F^2_h$ with $\|F\|_{\mathcal F^2_h}=1$. Equality in \eqref{eq:logsob} holds if and only if $F(z) = c e^{\beta z-\frac{h}{2\pi}|\beta|^2}$ for some $\beta,c\in\C$ with $|c|=1$.

The log-Sobolev inequality \eqref{eq:logsob} is closely related to an improved hypercontractivity property found by Janson \cite{Ja}; see also \cite{Ja2,Ca,Zh}. By a standard differentiation argument, which appears for instance in \cite{Lu}, this improved hypercontractivity implies \eqref{eq:logsob}. (As an aside, we mention that while the cases of equality are understood for the hypercontractivity assertion, we do not see how this yields the cases of equality in \eqref{eq:logsob}, in contrast to what is claimed in \cite{Lu}. We will obtain the characterization of equality by other means below.) We also note that the reverse, namely going from a log-Sobolev inequality to an improved hypercontractivity statement is studied in \cite{Gr} in the context of complex manifolds.

The log-Sobolev inequality \eqref{eq:logsob} for entire functions should be compared with its ordinary form without analyticity. The latter looks similar, but with $|\partial_z F|^2$ replaced by $|\nabla F|^2$. Since $|\partial_z F|^2=\frac12|\nabla F|^2$ for holomorphic functions (by the Cauchy--Riemann equations), we see that \eqref{eq:logsob} is an improvement, by a factor of two, of the ordinary log-Sobolev inequality.

As a corollary of our main result, we will now prove a quantitative form of \eqref{eq:logsob}. A similar quantitative form of the ordinary log-Sobolev inequality was shown recently in \cite{DoEsFiFrLo}, but we do not see how the method of proof employed there can be extended to entire functions. For the large literature on quantitative versions of the ordinary log-Sobolev inequality with different measures of the distance to the optimizers we refer to the bibliography in \cite{DoEsFiFrLo}.

\begin{theorem}
	There is a $c_*>0$ such that, for all $F\in\mathcal F^2_h$ with $\|F\|_{\mathcal F^2_h}=1$,
	\begin{equation}\label{eq aggiunto}
	\frac{h}{\pi} \int_\C |\partial_z F|^2\,d\mu_h(z)  \geq \int_\C |F(z)|^2 \ln |F(z)|^2 \,d\mu_h(z)
	+ c_* \inf_{\beta\in\C,\, |c|=1} \| F - c e^{\beta z-\frac{h}{2\pi}|\beta|^2} \|_{\mathcal F^2_h}^2 \,.
	\end{equation}
	The constant $c_*$ coincides with the constant in Theorem \ref{wehrlintro}.
\end{theorem}

The argument that follows is inspired by the proof of the sharp Lieb--Wehrl inequality in \cite{Lu}. There the log-Sobolev inequality \eqref{eq:logsob} is used to deduce the Lieb--Wehrl inequality. In contrast, we use our quantitative version of the Lieb--Wehrl inequality to obtain a quantitative version of the log-Sobolev inequality. The underlying idea is that the Lieb--Wehrl inequality is equivalent to the log-Sobolev inequality for entire functions.

\begin{proof}
	By scaling it suffices to prove the theorem for a single value of $h>0$. We choose $h=1$ and abbreviate $\mathcal F^2=\mathcal F^2_1$ and $\mu=\mu_1$. 
	
	In the proof of \eqref{eq aggiunto} we can of course assume that $\partial_z F\in \mathcal F^2$.  	
	We note that if $F\in \mathcal F^2$ and $\partial_z F\in \mathcal F^2$ we have
	$$
	- \frac{1}{\pi} \int_\C |\partial_z F|^2\,d\mu(z) = - \pi \int_\C |z|^2 |F|^2\,d\mu(z) + \int_\C |F|^2\,d\mu(z) \,.
	$$
	Indeed, this follows by integrating by parts twice, with the measure $\mu$ replaced by $e^{-(\pi+\varepsilon) |z|^2}dA(z)$, $\varepsilon>0$, and then letting $\varepsilon \to 0$, using the monotone convergence theorem. 
	
Thus, the inequality in the theorem is equivalent to the inequality
	$$
	- \int_\C |F(z)|^2 \ln( e^{-\pi |z|^2} |F(z)|^2)\,d\mu(z) - 1 \geq c_* \inf_{\beta\in\C,\, |c|=1} \| F - c e^{\beta z-\frac{1}{2\pi}|\beta|^2} \|_{\mathcal F^2}^2 \,.
	$$
	for all $F\in\mathcal F^2$ with $\|F\|_{\mathcal F^2}=1$.
	
	Recall that the operator $\mathcal V$ was defined in the introduction. For $f\in L^2(\R)$, we can write
	$$
	\mathcal Vf(x_0,\omega_0) = e^{-\frac\pi2(x_0^2+\omega_0^2)} e^{-i\pi x_0\omega_0} F(x_0-i\omega_0)
	$$
	with
	$$
	F(z) := 2^\frac14 \int_\R e^{2\pi z x-\frac\pi 2 z^2-\pi x^2} f(x)\,dx
	\qquad\text{for all}\ z\in\C \,.
	$$
	It is well known and easy to see that $F$ is entire and that $\|F\|_{\mathcal F^2} = \|f\|_{L^2(\R)}$. Moreover, one can show that the map $L^2(\R)\ni f\mapsto F \in \mathcal F^2$ is onto. Note that $|\mathcal Vf(x_0,\omega_0)|^2 = e^{-\pi(x_0^2+\omega_0^2)} |F(x_0-i\omega_0)|^2$.
	
	Also, a simple computation shows that the $F=\Phi_{(x_0,\omega_0)}$ corresponding to $f=\phi_{(x_0,\omega_0)}$ is equal to
	$$
	\Phi_{(x_0,\omega_0)}(z) = e^{i\pi x_0\omega_0} e^{-\frac\pi2(x_0^2+\omega_0^2)} e^{\pi(x_0+i\omega_0)z} 
	\qquad\text{for all}\ z\in\C \,.
	$$
	Note that this function is of the form $c_\beta e^{\beta z-\frac1{2\pi}|\beta|^2}$ with $|c_\beta|=1$, that is, it is an optimizer for the log-Sobolev inequality. Moreover, as $(x_0,\omega_0)$ ranges through $\R^2$, this family of functions ranges over all optimizers (up to constant phases).
	
	Thus, the inequality in the theorem is equivalent to the inequality
	$$
	- \iint_{\R^2} |\mathcal Vf(x_0,\omega_0)|^2\ln |\mathcal Vf(x_0,\omega_0)|^2 \,dx_0\,d\omega_0 -1 \geq c_* \inf_{x_0,\omega_0\in\R,\, |c'|=1} \| f - c' \phi_{(x_0,\omega_0)} \|_{L^2(\R)}^2 \,.
	$$
	This is precisely the assertion of Theorem \ref{wehrlintro} for $\rho=|f\rangle\langle f|$.
\end{proof}


\section{Acknowledgments} F.~N.\ and P.~T.\ would like to thank Aleksei Kulikov and Joaquim Ortega-Cerd\`a for useful discussions on the topic of this paper. 


\bibliographystyle{amsalpha}

\end{document}